\documentclass[a4paper]{lipics-v2021}
 \pdfoutput=1
\EventEditors{Neeldhara Misra and Magnus Wahlstr\"{o}m}
\EventNoEds{2}
\EventLongTitle{18th International Symposium on Parameterized and Exact Computation (IPEC 2023)}
\EventShortTitle{IPEC 2023}
\EventAcronym{IPEC}
\EventYear{2023}
\EventDate{September 6--8, 2023}
\EventLocation{Amsterdam, The Netherlands}
\EventLogo{}
\SeriesVolume{285}
\ArticleNo{14}
 
\hideLIPIcs

\usepackage{todonotes}

\newcommand{\SB}{\{\,}
\newcommand{\SM}{\;{|}\;}
\newcommand{\SE}{\,\}}

\newcommand{\CCC}{\mathcal{C}}
\newcommand{\QQQ}{\mathcal{Q}}

\newcommand{\FFF}{\mathcal{F}}

\usepackage{boxedminipage}
\usepackage{xspace}
\usepackage{amssymb}
\usepackage{microtype}
\usepackage{graphicx}
     \usepackage{url}             \usepackage{booktabs}        \usepackage{amsfonts}        \usepackage{nicefrac}        \usepackage{microtype}       \usepackage{verbatim}

\newcommand{\Nat}{\mathbb{N}}

\newcommand{\cc}[1]{{\mbox{\textnormal{\textsf{#1}}}}\xspace}

\newcommand{\NP}{\cc{NP}}

\newcommand{\FPT}{\cc{FPT}}

\newcommand{\Weft}{{\cc{W}}}
\newcommand{\W}[1]{{\Weft}{{\normalfont{[#1]}}}}
\newcommand{\paraNP}{\cc{paraNP}}

\newcommand{\III}{\mathcal{I}}

\newcommand{\bigoh}{\mathcal{O}}

\newcommand{\probfont}[1]{\textnormal{\textsc{#1}}}

 \newcommand{\INCLUS}{\probfont{In-Clustering}}

\newcommand{\PAIRCLUS}{\probfont{Diam-Clustering}}

\newcommand{\blank}{{\small\square}}

\newcommand{\LDIAMCL}{\probfont{Large Diam-Cluster}}

\newcommand{\DISP}{\probfont{Pow-Hyp-IS}}
\newcommand{\DISPq}{\probfont{Pow-Hyp-IS-Completion}}

\newcommand{\yes}{\textsc{Yes}}

\def\ie{{\em i.e.}}

\newcommand{\pbDef}[3]{ \noindent
\begin{center}
\begin{boxedminipage}{0.98 \columnwidth}
#1\\[5pt]
\begin{tabular}{l p{0.83 \columnwidth}}
Input: & #2\\
Question: & #3
\end{tabular}
\end{boxedminipage}
\end{center}
}

\newcommand{\pbDefP}[4]{ \noindent
\begin{center}
\begin{boxedminipage}{0.98 \columnwidth}
#1\\[5pt]
\begin{tabular}{l p{0.80 \columnwidth}}
Input: & #2\\
Parameter: & #3\\
Question: & #4
\end{tabular}
\end{boxedminipage}
\end{center}
}

\newcommand{\HSET}{\Delta}

\newcommand{\HDIST}{\delta}
\newcommand{\DIAM}{\gamma}

\newcommand{\HN}[2]{N_{#2}({#1})}
 
\newcommand{\compG}{G}

     \usepackage{graphicx}

\title{From Data Completion to Problems on Hypercubes: A Parameterized
  Analysis of the Independent Set Problem}
    \titlerunning{From Data Completion to Problems on Hypercubes}

\author{Eduard Eiben}{Department of Computer Science, Royal Holloway, University of London, Egham, UK}{eduard.eiben@gmail.com}{https://orcid.org/0000-0003-2628-3435}{}

\author{Robert Ganian}{Algorithms and Complexity Group, TU Wien, Vienna, Austria}{rganian@gmail.com}{https://orcid.org/0000-0002-7762-8045}{Robert Ganian acknowledges support from Project No.\ 10.55776/Y1329 of the Austrian Science Fund (FWF).}

\author{Iyad Kanj}{School of Computing, DePaul University, Chicago, USA}{ikanj@cdm.depaul.edu}{https://orcid.org/0000-0003-1698-8829}{Iyad Kanj acknowledges support from DePaul University through URC grant 602061.}

\author{Sebastian Ordyniak}{University of Leeds, School of Computing, Leeds, UK}{sordyniak@gmail.com}{https://orcid.org/0000-0003-1935-651X}{Project EP/V00252X/1 of the Engineering and Physical Sciences Research Council (EPSRC).}

\author{Stefan Szeider}{Algorithms and Complexity Group, TU Wien,
  Vienna,
  Austria}{sz@ac.tuwien.ac.at}{https://orcid.org/0000-0001-8994-1656}{Stefan
  Szeider acknowledges support from Project No.\
  10.55776/P36420 of the Austrian Science Fund (FWF).}

\Copyright{Eduard Eiben, Robert Ganian, Iyad Kanj, Sebastian Ordyniak, Stefan Szeider}

\authorrunning{E.\ Eiben, R.\ Ganian, I.\ Kanj, S.\ Ordyniak, S.\ Szeider}  
 

\category{}
\relatedversion{}
\nolinenumbers

\begin{document}   
 \maketitle
\begin{abstract}
  Several works have recently investigated the parameterized
  complexity of data completion problems, motivated by their
  applications in machine learning, and clustering in particular. Interestingly,
  these problems can be equivalently formulated as classical graph
  problems on induced subgraphs of powers of partially-defined
  hypercubes.

  In this paper, we follow up on this recent direction by investigating the
  Independent Set problem on this graph class, which has been studied in the
  data science setting under the name Diversity.  We obtain a
  comprehensive picture of the problem's parameterized complexity and
  establish its fixed-parameter tractability w.r.t. the solution size
  plus the power of the hypercube.

  Given that several such FO-definable problems have been shown to be
  fixed-parameter tractable on the considered graph class, one may ask
  whether fixed-parameter tractability could be extended to capture
  all FO-definable problems. We answer this question in the negative
  by showing that FO model checking on induced subgraphs of
  hypercubes is as difficult as FO model checking on general graphs.
\end{abstract}
  \section{Introduction}
\label{sec:intro}

Recently, there has been an increasing interest in studying the parameterized complexity of clustering problems motivated by their applications in machine learning~\cite{fomin1,fomin3,fomin2,fomin7,fomin6,fomin8,eibenesa,eibenjcss,fomin4,fomin5,ganian,niedermeierradius,niedermeierdiameter}, particularly their applications to fundamental clustering problems~\cite{clusteringbook1,clusteringbook2,clusteringbook4,clusteringbook3}.  
In many of these clustering problems, we are given a set of $d$-dimensional vectors over the Boolean/binary domain, where the vectors are regarded as rows of a matrix. It is worth noting that due to the applications of such problems in incomplete-data settings, a number of past works on the topic also studied settings where some of the entries in these vectors are unknown~\cite{icml,eibenesa,eibenjcss,ganian,fomin8,cp10,cr09,ct10,ev13,hmrw14}.
The objective is to determine if these vectors (or, in the incomplete-data setting, their completions) satisfy some desirable clustering properties.
  Examples of such properties include admitting a partitioning into $k$ clusters each of diameter (or radius) at most $r$ (for some given $k, r \in \Nat$), or 
admitting a $k$-cluster of diameter (or radius) at most $r$, where the distance under consideration is typically the Hamming distance~\cite{charikar,frieze,eibenesa,eibenjcss,normclustering,lingashammingcenter,lingasapxclustering,gonzalez,GrammNiedermeierRossmanith03}; here, a $k$-cluster of diameter $r$ is a set of $k$ points which have pairwise distance of at most $r$.

As it turns out, many of these well-studied clustering problems can be formulated as classical graph problems on induced subgraphs of powers of the hypercube graph. For instance, 
finding a cluster of diameter at most $r \in \Nat$, for a given $r$,
is equivalent to the \textsc{Clique} problem defined on the subgraph of the $r$-th power of the hypercube that is induced by the subset of hypercube vertices corresponding to the given input vectors.
Similarly, partitioning the set of vectors into $k$ clusters each of diameter at most $r$, for some given $r, k \in \Nat$, is equivalent to the partitioning into $k$ cliques problem on the same graph class, whereas partitioning the set of vectors into clusters, each of radius at most $r$ with respect to some vector in the set, is equivalent to the $k$-dominating set problem on the same graph class described above. We remark that, to the best of our knowledge, this graph class is not a subclass of commonly studied graph classes and has not been considered in previous works pertaining to algorithmic upper or lower bounds for graph-theoretic problems.

\smallskip
\noindent \textbf{Contribution.}\quad
In this paper, we study the parameterized complexity of another classical graph problem defined on induced subgraphs of powers of the hypercube: the \textsc{Independent Set} problem.  In the context of data analytics, the problem arises when studying the
``diversity'' of a given set of vectors, a notion that can be viewed
as the opposite of minimising the number of clusters in a cluster
partitioning of the set of vectors (in fact, in the area of data analytics this is studied directly under the nomenclature
\emph{diversity} or dispersion~\cite{CeccarelloPPU17,gawrychowski_et_al,SacharidisMSPV18}). 
     More precisely, motivated by the aforementioned extensive interest in the analysis of incomplete data, we focus on the more general incomplete data setting. 
We refer to this problem as \DISPq: given a set of Boolean vectors with some missing entries and integers $k$ and $r$, the goal is to complete the missing entries so that the resulting set of vectors contains a subset $S$ of $k$ vectors such that the Hamming distance between each pair is at least $r+1$ (or to correctly determine that such a set does not exist).

The main contribution of this paper is a complete characterisation of the parameterized complexity of
\DISPq{} w.r.t.~the two parameters $k$ and $r$: we provide a fixed-parameter algorithm for \DISPq{} when parameterized by $k+r$, and
complement this positive result with intractability results for the
cases where any of these two parameters is dropped. 
In particular, we show that the problem is \NP-complete
already for $r=2$---that is, the problem is \emph{\paraNP{}-hard} parameterized by $r$, and \W{1}-hard parameterized by $k$ alone.  Interestingly, the \FPT{} result shows that the parameterized complexity of the problem is independent of any restrictions on the number or the structure of the missing entries in the input vectors---contrasting many of the previous results on clustering incomplete data~\cite{icml,eibenesa,eibenjcss,ganian}. We remark that even the fixed-parameter tractability of the problem in the complete data setting (i.e., where all entries are known) is non-obvious, but follows as an immediate corollary of our result.
 
For our final contribution, we revisit the observation that several of the complete-data clustering problems recently considered in the literature (e.g., see~\cite{eibenesa,eibenjcss}) reduce to well-known graph problems on the class of induced subgraphs of powers of the hypercube. 
 Since it was shown that all of these graph problems are fixed-parameter tractable when restricted to this graph class and the graph problems are
  expressible in First Order Logic (FO), a natural question to ask is whether these \FPT{} results can be generalised to any graph problem expressible in FO logic.
  We resolve this question in the negative.

\smallskip
\noindent \textbf{Related Work.}\quad
The problem of computing the diversity of a data set, which forms the underpinning of our study of \DISPq, has been studied in a variety of different contexts and settings. For instance, Ceccarello, Pietracaprina, Pucci and Upfal studied approximation algorithms for the problem~\cite{CeccarelloPPU17}. Gawrychowski, Krasnopolsky, Mozes, and Weimann obtained a linear-time algorithm for the problem when the data set is represented as a tree~\cite{gawrychowski_et_al}, improving upon the previous polynomial-time algorithm of Bhattacharya and Houle~\cite{BhattacharyaH99}. Sacharidis, Mehta, Skoutas, Patroumpas and Voisard provided heuristics for dynamic versions of the problem~\cite{SacharidisMSPV18}.

More broadly, there is extensive work on problems arising in the context of incomplete data. Hermelin and Rozenberg~\cite{hermelin} studied the {\sc Closest String with Wildcards} problem, which can be seen as the problem of finding a data completion and a center to a minimum-radius cluster containing all the data points. Koana, Froese and Niedermeier~\cite{niedermeierradius} recently revisited the earlier work of Hermelin and Rozenberg~\cite{hermelin} and obtained, among other results, a fixed-parameter algorithm for that problem parameterized by the radius plus the maximum number of missing entries per row; see also the related work of the same authors~\cite{niedermeierdiameter}.
Eiben et al.~considered a number of different clustering problems in the presence of incomplete data~\cite{EibenGKOS21,eibenesa}, and a subset of these authors previously investigated the fundamental \textsc{Matrix Completion} problem in the same setting~\cite{icml}.
The parameterized complexity of $k$-means clustering on incomplete data was investigated by Eiben et al.~\cite{fomin8} and Ganian et al.~\cite{ganian}.

\section{Preliminaries}
 \label{section:prelims}
\newcommand{\compGq}{\compG}

\newcommand{\HSETq}{\HSET}
 \newcommand{\HDISTq}{\HDIST}
 \newcommand{\DIAMq}{\DIAM}
\newcommand{\MDIAMq}{\DIAM_{\max}}
 \newcommand{\HNq}{\HN}
 
 \paragraph*{Problem Terminology and Definition}
Let $\vec{a}$ and $\vec{b}$ be two vectors in $\{0,1,\blank\}^d$, where $\blank$ is used to represent coordinates whose value is unknown (\ie, missing entries). We denote  
by $\HSET(\vec{a},\vec{b})$ the set of coordinates in which
$\vec{a}$ and $\vec{b}$ are guaranteed to differ, \ie,
$\HSET(\vec{a},\vec{b})=\SB i
\SM (\vec{a}[i]=1 \wedge \vec{b}[i]=0)\vee (\vec{a}[i]=0 \wedge \vec{b}[i]=1) \SE$, and we denote by
$\HDIST(\vec{a},\vec{b})$ the \emph{Hamming distance} between
$\vec{a}$ and $\vec{b}$ measured only between known entries, \ie,
$|\HSET(\vec{a},\vec{b})|$. Moreover, for a subset $D' \subseteq [d]$
of coordinates, we denote by $\vec{a}[D']$ the vector $\vec{a}$
restricted to the coordinates in $D'$.

Let $M\subseteq \{0,1\}^d$ and let $[d]=\{1,\dots,d\}$. For a vector $\vec{a} \in M$ and $t \in \Nat$, we
denote by $\HN{\vec{a}}{t}$ the
\emph{$t$-Hamming neighbourhood}
 of $\vec{a}$, \ie, the set $\SB
\vec{b} \in M \SM \HDIST(\vec{a},\vec{b})\leq t\SE$ and by $\HN{M}{t}$ the
set $\bigcup_{\vec{a} \in M}\HN{\vec{a}}{t}$.
   We say that $M^*\subseteq \{0,1\}^d$ is a \emph{completion} of $M\subseteq \{0,1,\blank\}^d$ if there is a bijection $\alpha: M \rightarrow M^*$ such that for all $\vec{a}\in M$ and all $i\in [d]$ it holds that either $\vec{a}[i]=\blank$ or $\alpha(\vec{a})[i]=\vec{a}[i]$.

We now proceed to give the formal definition of the problem under consideration:

\pbDef{\DISPq{}}{A set $M$ with elements from $\{0,1,\blank\}^d$ and $k, r \in \Nat$.}{Is there a completion $M^*$ of $M$ and a subset $S$ of $M^*$ with $|S|= k$ such that, for any two vectors $a, b \in S$, we have $\delta(a, b) \geq r+1$?}

Observe that in a matrix representation of the above problem, we can
represent the input matrix as a \emph{set} of vectors where each row
of the matrix corresponds to one element in our set.  

We remark that even though the statements are given in the form of
decision problems, all tractability results presented in this paper are constructive and the associated algorithms can also output a solution (when it exists) as a witness, along with the decision.
 In the case where we restrict the input to vectors over $\{0,1\}^d$ (\ie, where all entries are known), we omit ``-\textsc{Completion}'' from the problem name.

 \paragraph*{Parameterized Complexity} The basic motivation behind parameterized complexity is to find a parameter that describes the structure of
the problem instance such that the combinatorial explosion can be confined to
this parameter. More formally, a {\it parameterized problem} $Q$ is a subset of $\Omega^* \times
\mathbb{N}$, where $\Omega$ is a fixed finite alphabet. Each instance of $Q$ is a pair $(I, \kappa)$, where $\kappa \in \Nat$ is called the {\it
parameter}. A parameterized problem $Q$ is
{\it fixed-parameter tractable} (\FPT)~\cite{FlumGrohe06,DowneyFellows13,CyganFKLMPPS15}, if there is an
algorithm, called an {\em \FPT-algorithm},  that decides whether an input $(I, \kappa)$
is a member of $Q$ in time $f(\kappa) \cdot |I|^{\bigoh(1)}$, where $f$ is a computable function and $|I|$ is the input instance size.  The class \FPT{} denotes the class of all fixed-parameter
tractable parameterized problems.

A parameterized problem $Q$
is {\it \FPT-reducible} to a parameterized problem $Q'$ if there is
an algorithm, called an \emph{\FPT-reduction}, that transforms each instance $(I, \kappa)$ of $Q$
into an instance $(I', \kappa')$ of
$Q'$ in time $f(\kappa)\cdot |I|^{\bigoh(1)}$, such that $\kappa' \leq g(\kappa)$ and $(I, \kappa) \in Q$ if and
only if $(I', \kappa') \in Q'$, where $f$ and $g$ are computable
functions.  
Based on the notion of \FPT-reducibility, a hierarchy of
parameterized complexity, {\it the \cc{W}-hierarchy} $=\bigcup_{t
\geq 0} \W{t}$, where $\W{t} \subseteq \W{t+1}$ for all $t \geq 0$, has
been introduced, in which the $0$-th level \W{0} is the class {\it
\FPT}. The notions of hardness and completeness have been defined for each level
\W{$i$} of the \cc{W}-hierarchy for $i \geq 1$ \cite{DowneyFellows13,CyganFKLMPPS15}. It is commonly believed that $\W{1} \neq \FPT$ (see \cite{DowneyFellows13,CyganFKLMPPS15}). The
\W{1}-hardness has served as the main working hypothesis of fixed-parameter
intractability.
   A problem is \emph{\paraNP{}-hard} if it is \NP-hard for a constant value of the parameter~\cite{FlumGrohe06}.

\paragraph*{Sunflowers} A \emph{sunflower} in a set family $\FFF$ is a subset $\FFF' \subseteq \FFF$ such that all pairs of elements in $\FFF'$ have the same intersection.

\begin{lemma}[\cite{Erdos60,FlumGrohe06}]\label{lem:SF}
  Let $\FFF$ be a family of subsets of a universe $U$, each of cardinality exactly
  $b$, and let $a \in \mathbb{N}$. If $|\FFF|\geq b!(a-1)^{b}$, then $\FFF$
  contains a sunflower $\FFF'$ of cardinality at least $a$. Moreover,
  $\FFF'$ can be computed in time polynomial in $|\FFF|$.
\end{lemma}
                          \section{The Parameterized Complexity of \DISPq}
\label{sec:dispersion}

Our aim for \DISPq\ is to establish fixed-parameter tractability parameterized by $k+r$ (\ie, regardless of the structure or number of missing entries).
 As our first step, we show that all rows in an arbitrary instance $(M,k,r)$ can be, w.l.o.g., assumed to contain at most $\bigoh(k\cdot r)$ many $\blank$'s.  

Next, we observe that if $M$ is sufficiently large and the $r$-Hamming neighbourhood of each vector is upper-bounded by a function of $k+r$, then---since the number of $\blank$'s is bounded---$(M,k,r)$ is a YES-instance. The argument here is analogous to the classical argument showing that \textsc{Independent Set} is trivial on large bounded-degree graphs.

On a high level, we would now like to find and remove an ``irrelevant vector''
from $M$---since here the number
of $\blank$'s on \emph{every} row is bounded, any instance reduced in
this way to only contain a bounded number of vectors can be solved via
a brute-force fixed-parameter procedure. However, finding an
irrelevant vector is rather challenging, primarily because the
occurrence of $\blank$'s is not restricted. Instead, we develop a more
powerful set representation $\FFF'$ for vectors in the instance which
also uses elements to keep track of the presence of $\blank$'s in the
neighbours of $\vec{v}$. We can then apply the Sunflower Lemma to find
a sufficiently-large sunflower in $\FFF'$, and in the core of the
proof we argue that (1) such a sunflower consists of at most a bounded
number of ``important petals'' (which can be identified in polynomial
time), and (2) any petal that is not important represents an
irrelevant vector.

 \subsection{Dealing with Unstructured Missing Data}
In this subsection, we design an algorithm for \DISPq\ which remains efficient even when the number and placement of unknown entries is not explicitly restricted on the input.

We begin with a simple lemma that allows us to deal with vectors (\ie, rows) with a large number of missing entries.    For brevity, let a $k$-diversity set be a set containing $k$ vectors which have pairwise Hamming distance at least $r+1$.

\begin{lemma}
\label{lem:ISmanyq}
Let $\III=(M,k,r)$ be an instance of \DISPq{} where $k\geq 1$ and let $\vec{v}\in M$ be a vector containing more than $(k-1)\cdot(r+1)$-many $\blank$'s. Then $\III$ is a \textsc{YES}-instance if and only if $\III'=(M\setminus \{\vec{v}\},k-1,r)$ is a \textsc{YES}-instance. Moreover, a completion and $k$-diversity set for $\III$ can be computed from a completion and $(k-1)$-diversity set for $\III'$ in linear time.
\end{lemma}

\begin{proof}
The forward direction is trivial: for any completion $M^*$ of $M$ and $k$-diversity set $S$ in $M^*$, we can obtain a $(k-1)$-diversity set and completion for $\III'$ by simply removing $\vec{v}$ from $M^*$ and $S$.

For the backward direction, consider a completion ${M'}^*$ of $M'=M\setminus \vec{v}$ and a $(k-1)$-diversity set $S=\{\vec{s_1},\dots,\vec{s_{k-1}}\}$ in $M'^*$. Let us choose an arbitrary set $C$ of $(k-1)\cdot(r+1)$ coordinates in $\vec{v}$ that all contain $\blank$, and let us then partition $C$ into $k$-many subsets $\alpha_1,\dots,\alpha_k$ each containing precisely $r+1$ coordinates. Now consider the vector $\vec{v}^*$ obtained from $\vec{v}$ as follows:
\begin{itemize}
\item for each $i\in [k-1]$ and every coordinate $j\in \alpha_i$, set $\vec{v}^*[j]$ to the opposite value of $\vec{s_i}[j]$ (\ie, $\vec{v}^*[j]=1$ if and only if $\vec{s_i}[j]=0)$;
\item for every other coordinate $j$ of $\vec{v}^*$, we set $\vec{v}^*[j]=\vec{v}[j]$ if $\vec{v}[j]\neq \blank$ and $\vec{v}^*[j]=0$ otherwise.
\end{itemize}

Clearly, $M^*=M'^*\cup \{\vec{v}^*\}$ is a completion of $M$. Moreover, since $\vec{v}^*$ differs from each vector in $S$ in at least $r+1$ coordinates, $S\cup\{\vec{v}^*\}$ is a $k$-diversity set in $M^*$. 
\end{proof}

Next, we show that instances which are sufficiently large and where each vector only ``interferes with'' a bounded number of other vectors are easy to solve.
Technically, let $$\zeta(k,r)=3^{(k-1)\cdot(r+1)}\cdot \sum_{\alpha\in [(k-1)\cdot (r+1)+r]}\big(\alpha! \cdot ((k-1)\cdot 2\cdot (3(k-1)\cdot (r+1)+2r))^\alpha\big)$$

 be the exact meaning of ``sufficiently large'' here; for brevity, note that $\zeta(k,r)\in (kr)^{\bigoh(kr)}$.

\begin{lemma}
\label{lem:ISdegree}
Let $\III=(M,k,r)$ be an instance of \DISPq{}. If
$|M|\geq k\cdot \zeta(k,r)$ and
$|\HN{\vec{v}}{r}|< \zeta(k,r)$
for every $\vec{v}\in M$, then a $k$-diversity set in $\III$ can be found in polynomial time.
\end{lemma}
\begin{proof}
One can find a solution to $\III$ by iterating the following greedy procedure $k$ times: choose an arbitrary vector $\vec{v}$, add it into a solution, and delete all other vectors with Hamming distance at most $r$ from $\vec{v}$. By the bound on $|\HN{\vec{v}}{r}|$, each choice of $\vec{v}$ will only lead to the deletion of at most
$\zeta(k,r)$ vectors from $M$.
 Moreover, since $\HDIST$ measures the Hamming distance only between known entries, \emph{any} completion of the missing entries can only increase (and never decrease) the Hamming distance between vectors. Hence, the size of $M$ together with the bounded size of the Hamming neighbourhood of $\vec{v}$ guarantee that this procedure will find a solution of cardinality $k$ in $\III$ which will remain valid for every completion of $M$. 
\end{proof}

We can now move on to the main part of the proof: a procedure which either outputs a solution outright or finds an irrelevant vector.

\begin{lemma}
\label{lem:ISprune}
Let $\III=(M,k,r)$ be an instance of \DISPq{} such that $|\HN{\vec{v}}{r}|\geq \zeta(k,r)$ for some vector $\vec{v}\in M$ and such that each vector in $M$ contains at most $(k-1)\cdot (r+1)$ $\blank$'s. There is a polynomial-time procedure that finds a vector $\vec{f}\in M$ satisfying the following properties: \begin{itemize}
 \item $(M,k,r)$ is a YES-instance if and only if $\III' =(M\setminus \{\vec{f}\},k,r)$ is a YES-instance, and
 \item A completion and diversity set for $\III$ can be computed from a solution and diversity set for $\III'$ in linear time.
 \end{itemize}
\end{lemma}

\begin{proof}
We will begin by constructing a set system over the neighbourhood of $\vec{v}$. Let $Z=\SB z\in [d] \SM \vec{v}[z]=\blank \SE$ be the set of coordinates where $\vec{v}$ is incomplete. Clearly, since 
$$|\HN{\vec{v}}{r}|\geq 3^{(k-1)\cdot(r+1)}\cdot \sum_{\alpha\in [(k-1)\cdot (r+1)+r]}\big(\alpha! \cdot ((k-1)\cdot 2\cdot (3(k-1)\cdot (r+1)+2r))^\alpha\big)$$
  and $|Z|\leq (k-1)\cdot (r+1)$,
we can find a subset $N\subseteq \HN{\vec{v}}{r}$ of vectors whose cardinality is at least 
$\sum_{\alpha\in [(k-1)\cdot (r+1)+r]}\big(\alpha! \cdot ((k-1)\cdot 2\cdot (3(k-1)\cdot (r+1)+2r))^\alpha\big)$
  such that all vectors in $N$ are precisely the same on the coordinates in $Z$, \ie, $\forall \vec{x},\vec{y}\in N: \forall z\in Z: \vec{x}[z]=\vec{y}[z]$.

Now, let $F$ be a set containing $2$ elements for each coordinate $j\in [d]\setminus Z$ of vectors in $M$: the element $\blank_j$ and the element $D_j$. We construct a set system $\FFF$ over $F$ as follows. For each vector $\vec{x}\in N$, we add a set $\hat{x}$ to $\FFF$ that contains:
\begin{itemize}
\item $\blank_j$ if and only if $\vec{x}[j]=\blank$, and
\item $D_j$ if and only if $\vec{x}[j]\neq \blank$ and $\vec{x}[j]\neq \vec{v}[i]$.
\end{itemize}

Observe that, since $\vec{x}$ contains at most $(k-1)\cdot (r+1)$ $\blank$'s by assumption and since $\vec{x}$ differs from $\vec{v}$ in at most $r$-many completed coordinates, every set in $\FFF$ has cardinality at most $(k-1)\cdot (r+1)+r$. 
This means that there exists $\alpha\in [(k-1)\cdot (r+1)+r]$ such that there are at least $\alpha! \cdot ((k-1)\cdot 2\cdot (3(k-1)\cdot (r+1)+2r))^\alpha$ vectors $\vec{x}\in N$ such that $|\hat{x}|=\alpha$. This means we can apply Lemma~\ref{lem:SF} to find a sunflower $\FFF'$ in $\FFF$ of cardinality at least $(k-1)\cdot 2\Big(3(k-1)\cdot (r+1)+2r\Big)+1$; for ease of presentation, we will identify the elements of $\FFF'$ with the vectors they represent.
  Let $\vec{f}$ be an arbitrarily chosen vector from $\FFF'$; we claim that $\vec{f}$ satisfies the properties claimed in the lemma, and to complete the proof it suffices to establish this claim.

The backward direction is trivial: if $\III'$ is a YES-instance then clearly $\III$ is a YES-instance as well. It is also easy to observe that a completion and diversity set for $\III$ can be computed from a solution and diversity set for $\III'$ in linear time (adding a vector does not change the validity of a solution). What we need to show is that if $\III$ is a YES-instance, then so is $\III'$ (\ie, $(M\setminus \{\vec{f}\}, k, r)$); moreover, this final claim clearly holds if $\III$ admits a solution that does not contain $\vec{f}$.

So, assume that $M$ admits a completion $M^*$ which contains a $k$-diversity set $S=\{\vec{f},\vec{s_1},\dots,\vec{s_{k-1}}\}$. Let $C$ be the core of the sunflower $\FFF'$, and note that all vectors in $\FFF'$ have precisely the same content in the coordinates in $C$.

\paragraph*{Finding a replacement for $\vec{f}$}
We would now like to argue that, for some completion which we will
define later, $\FFF'$ contains a vector that can be used to replace
$\vec{f}$ in the solution.
 
Let $\vec{s_i}\in S$ be an arbitrary vector. First, let us consider the case that, in $M$, $\vec{s_i}$ differs from $\vec{v}$ in more than $3(k-1)\cdot (r+1)+2r$ coordinates (\ie, $\vec{v}[j]\neq \vec{s_i}[j]$ in $M$ for at least $3(k-1)\cdot (r+1)+2r$ choices of $j$).
 Then \emph{every} vector in $\FFF'$ will have Hamming distance greater than $r$ from $\vec{s_i}$ \emph{regardless of the completion}.

Indeed, for every vector $f'\in \FFF'$ there are at most $3(k-1)\cdot (r+1)$  coordinates $j$ such that at least one of $\vec{v}[j]$, $\vec{s_i}[j]$, $\vec{f'}$, meaning that there are at least $2r$ \emph{other} coordinates where $\vec{v}$ differs from $\vec{s_i}$ and which are guaranteed to be complete---and since $\HDIST(\vec{f'},\vec{v})\leq r$, $\vec{f'}$ it must hold that $\HDIST(\vec{f'},\vec{s_i})>r$ (by the triangle inequality). Hence indeed every vector in $\FFF'$ must have distance at least $r+1$ from $\vec{s_i}$, and in this case we will create a set $S_i=\emptyset$ (the meaning of this will become clear later).

Now, consider the converse case, \ie, that $\vec{s_i}$ differs from $\vec{v}$ in at most $3(k-1)\cdot (r+1)+2r$ coordinates. We may now extend the set system over $F$ by adding a set representation of $\vec{s_i}$, specifically by adding a set $Q_i$ such that $\{\blank_j,D_j\}\subseteq Q_i$ if $\vec{s_i}[j]\neq \vec{v}[j]$ (note that since $\vec{v}[j]\neq \blank$, this also includes the case $\vec{s_i}[j]=\blank$) and otherwise $\{\blank_j,D_j\}\cap Q_i=\emptyset$.
Observe that $|Q_i|\leq 2\cdot (3(k-1)\cdot (r+1)+2r)$, and in particular $Q_i\setminus C$ intersects with at most $2\cdot (3(k-1)\cdot (r+1)+2r)$ elements of $\FFF'$. Let $S_i$ be the set of all such elements, \ie, sets in $\FFF'$ which have a non-empty intersection with $Q_i$ outside of the core (formally, these are sets of $\FFF'$ that do intersect $Q_i\setminus C$). 
  
Observe that by the construction of $\FFF'$, there must exist at least one set in the sunflower that does not lie in any $S_i$. To conclude the proof, we will show that there is a completion $M'^*$ of $M'$ such that any arbitrarily chosen vector $\vec{f'}$ in the non-empty set $\FFF'\setminus (\{\vec{f}\}\cup \bigcup_{i\in [k-1]}S_i)$ can replace $\vec{f}$ in the $k$-diversity set~$S$.

\paragraph*{Arguing Replaceability}
Consider a new completion $M'^*$ of $M\setminus \vec{f}$ obtained as follows:
\begin{itemize}
\item For each vector $\vec{w}\in \FFF'\setminus S$, we complete
\begin{enumerate}
\item the $\blank$'s in $C\cup Z$ precisely in the same way as $\vec{f}$, and
\item for every other $\blank$ at coordinate $j$, we set $\vec{w}[j]=-(\vec{v}[j]-1)$ (\ie, to the opposite of $\vec{v}$ -- recall that $\vec{v}[j]\neq \blank$ since $j\not \in Z$); and
\end{enumerate}
\item all other $\blank$'s in all other vectors in $M\setminus \vec{f}$ are completed in precisely the same way as in $M^*$.
 \end{itemize}

Since $M'^*$ precisely matches $M^*$ on all vectors in $S\setminus \vec{f}$, it follows that $S\setminus \vec{f}$ is a $(k-1)$-diversity set in $M'^*$. Moreover, consider for a contradiction that $\HDIST(\vec{f'},\vec{s_i})\leq r$ for some $\vec{s_i}\in S$ \emph{after completion}, \ie, in $M'^*$. Then clearly $\vec{s_i}$ could not differ from $\vec{v}$ in more than  $3(k-1)\cdot (r+1)+2r$ coordinates in $M'$, since---as we already argued---in this case every vector in $\FFF'$ will have Hamming distance greater than $r$ from $\vec{s_i}$ regardless of the completion.

Hence, we must be in the case where $\vec{s_i}$ differed from $\vec{v}$ in at most $3(k-1)\cdot (r+1)+2r$ coordinates in $M'$. Now consider how $\HDIST(\vec{f'},\vec{s_i})$ differs from $\HDIST(\vec{f},\vec{s_i})$. First of all, there is no difference between these two distances on the coordinates in $Z\cup C$ due to our construction of $M'^*$ and choice of $N$. For the remaining coordinates, we will consider separately the set $X$ of coordinates in the petals of $\vec{f}$ and $\vec{f'}$ (\ie, the set $\SB j\in [d]\setminus (Z\cup C) \SM \vec{f}[j]\neq \vec{v}[j] \vee \vec{f'}[j]\neq \vec{v}[j]\SE$), and the set $Y=[d]\setminus (C\cup Z\cup X)$ of all remaining coordinates. It follows that $\vec{v}[j]=\vec{f}[j]=\vec{f'}[j]$ for all coordinates $j\in Y$, and hence there is no difference between the two distances on these coordinates either.

So, all that is left is to consider the difference between $\HDIST(\vec{f'},\vec{s_i})$ and $\HDIST(\vec{f},\vec{s_i})$ on the coordinates in $X$; it will be useful to recall, that(unlike for the sets $Q_i$, the construction of $\FFF'$ guarantees that each coordinate occurs at most once in $\vec{f}$ and also at most once in $\vec{f'}$, and that $\alpha$ is the size of each set in the sunflower $\FFF'$. 
Among the coordinates in $X$, $\vec{f}$ can only differ from $\vec{s_i}$ in \emph{at most} $\alpha-|C|$ many coordinates---notably in the coordinates of its own petal---because the coordinates in the petal of $\vec{f'}$ do not intersect with $Q_i$. 
On the other hand, our construction guarantees that $\vec{f'}$ differs from $\vec{s_i}$ in \emph{at least} $\alpha-|C|$ coordinates in $X$; more precisely, on all coordinates in the petal of $\vec{f'}$, since on these coordinates (1) $\vec{s_i}$ is equal to $\vec{v}$ and (2) $\vec{f'}$ differs from $\vec{v}$.

In summary, we conclude that $\HDIST(\vec{f'},\vec{s_i})\geq \HDIST(\vec{f},\vec{s_i})$ and hence $(S\setminus \{\vec{f}\})\cup \{\vec{f'}\}$ is a $k$-diversity set in $M'^*$, as claimed. 
\end{proof}

We can now establish our main result for \DISPq.

\begin{theorem}
\label{the:IS}
\DISPq{} is fixed-parameter tractable parameterized by $k+r$.
\end{theorem}

\begin{proof}
The algorithm proceeds as follows. Given an instance $\III=(M,k,r)$ of \DISPq{}, it first checks whether $M$ contains a vector with more than $(k-1)\cdot (r+1)$ $\blank$'s; if yes, it applies Lemma~\ref{lem:ISmanyq} and restarts on the reduced instance.
Second, it checks whether $|M|\geq k\cdot \zeta(k,r)$; if not, it uses the fact that the number of $\blank$'s and the number of rows is bounded by a function of the parameter to find a completion and a $k$-diversity set in $\III$ (or determine that one does not exist) by brute force.

Third, it checks whether each vector $\vec{v}$ satisfies $|\HN{\vec{v}}{r}|< \zeta(k,r)$; if yes, then it solves $\III$ by invoking Lemma~\ref{lem:ISdegree}. Otherwise, it invokes Lemma~\ref{lem:ISprune} to reduce the cardinality of $M$ by $1$ and restarts. If the algorithm eventually terminates with a ``NO'', then we know that the initial input was a NO-instance; otherwise, it will output a solution which can be transformed into a solution for the original input by the used lemmas. 
\end{proof}

\subsection{Lower Bounds}\label{sec:lb}
 
\begin{theorem}\label{the:dispersionlowerbound}
  \DISP{} is \NP-complete and \Weft\emph{[1]}-hard parameterized by $k$.
\end{theorem}

\begin{proof}
  We prove both \NP-hardness and \W{1}-hardness results by giving a
  polynomial-time \FPT{} reduction from {\sc Independent Set} (IS),
  which is \W{1}-hard~\cite{DowneyFellows13}.

  Let $(G, k)$ be an instance of IS, where $V(G)=\{v_1, \ldots, v_n\}$,
  and let $m=E(G)$.  Fix an arbitrary ordering ${\cal O}=(e_1, \ldots,
  e_m)$ of the edges in $E(G)$.
  
  For each vertex $v_i \in V(G)$, define a vector $\vec{a_i} \in \{0,
  1\}^m$ by setting $\vec{a_i}[j]=1$ if $v_i$ is incident to $e_j$ and
  $\vec{a_i}[j]=0$ otherwise. Now expand the set of coordinates of
  these vectors by adding to each of them $n(n-1)$ new coordinates,
  $n-1$ coordinates for each $v_i$, $i \in [n]$; we refer to the $n-1$
  (extra) coordinates of $v_i$ as the ``private'' coordinates of
  $v_i$. For each $v_i$, $i \in [n]$,  set $n-1-deg(v_i)$ many
  coordinates among the private coordinates of $v_i$ to 1, and all
  other new coordinates of $v_i$ to 0. Let $M=\{\vec{a_i} \mid i \in
  [n]\}$ be the set of expanded vectors, where $\vec{a_i} \in \{0,
  1\}^{m+n(n-1)}$, for $i \in [n]$. The reduction from IS to \DISP{}
  produces the instance $\III=(M, k, 2n-4)$ of \DISP{}; clearly, this
  reduction is a polynomial-time \FPT-reduction.
  
  Observe that, for any two distinct vertices $v_i, v_j \in V(G)$,
  $\HDIST(\vec{a_i},\vec{a_j}) =2n-2$ if $v_i$ and $v_j$ are
  nonadjacent and $\HDIST(\vec{a_i},\vec{a_j}) =2n-4$ if $v_i$ and
  $v_j$ are adjacent.
  
  The proof that $(G, k)$ is a \yes-instance of IS iff $(M, k, 2n-4)$
  is a \yes-instance of \DISP{} is now straightforward. 
\end{proof}

\begin{theorem}\label{the:dispersionlowerbound-rad}
  \DISP{} is \NP-complete even when $r=2$.
\end{theorem}
\begin{proof}
  We reduce from the \textsc{Independent Set} problem (which
  is \NP-complete). Let $(G,k)$ be an instance
  of \textsc{Independent Set} and let $G'$ be the graph obtained from
  $G$ after subdividing every edge exactly twice. We first observe
  that $G$ has an independent set of size at least $k$ if and only if
  $G'$ has an independent set of size at least $|E(G)|+k$. This is because
  if $I \subseteq V(G)$ is an independent set of $G$, then we can add one
  of the subdivision vertices for every edge of $G$ because $I$ does not contain
  both endpoints of an edge. On the other hand, if $I \subseteq V(G')$ is an
  independent set of $G'$, then we can assume without loss of generality that $I$
  does not contain both endpoints of an edge in $G$ because we could
  easily transform $I$ into an independent set of the same size by
  replacing one of the endpoints of such an edge with a subdivided
  vertex.

  Next we construct an instance $\III=(M,|E(G)|+k,2)$ of \DISP{} in
  polynomial-time such that $G'$ has an independent set of size at
  least $|E(G)|-k$ if and only if $\III$ is a \yes-instance. We set
  $d=2|V(G)|$ and obtain $M$ as follows. Let
  $V(G)=\{v_1,\dots,v_n\}$. For every $v_i \in V(G)$, we
  add the vector $\vec{v}_i$ that is $1$ at the two coordinates $i$
  and $i+1$ and otherwise $0$. Moreover, for every $e=v_iv_j\in E(G)$,
  we add the vector $e^1$ that is $1$ at the coordinates $i$, $i+1$,
  and $j$ and the vector $e^2$ that is $1$ at the coordinates $j$, $j+1$,
  and $i$. This completes the construction of $\III$. The equivalence
  now follows because two vectors in $M$ have distance at most $r=2$
  if and only if their corresponding vertices in $G'$ are adjacent;
  here $e^1$ and $e^2$ correspond to the two subdivision vertices on
  the edge $e$. 
\end{proof}

\section{On Graph Problems on Induced Subgraphs of the Hypercubes}
\label{sec:implications}

In this section, we discuss the implications of the results in
the previous section for fundamental problems
defined on induced subgraphs of powers of the hypercube graph.

In particular, the $d$-dimensional hypercube graph is the graph $Q_d$
whose vertex set is the set of all Boolean $d$-dimensional vectors,
and two vertices are adjacent if and only if their two vectors differ
in precisely $1$ coordinate. We can then define the class
$\QQQ^r_d$  as the class of all graphs that are induced subgraphs of the $r$-th
power of $Q_d$. We note that, in line with the commonly used
definition of hypercube graphs~\cite{Dvorak2007,HararyHayessurvey}, we
consider the vertices in $\QQQ^r_d$ to be vectors and hence every
graph $G\in \QQQ^r_d$ contains an explicit characterisation of its
vertices as vectors.

In this setting, it is straightforward to observe that \DISP{} is
precisely the {\sc Independent Set} problem on $\QQQ_d^r$. Moreover,
the clustering problems \INCLUS{}, \PAIRCLUS, and \LDIAMCL\ considered
in~\cite{eibenesa,eibenjcss} are precisely the {\sc Dominating Set}, {\sc Partition Into Cliques}, and {\sc
Clique} problems, respectively, on $\QQQ_d^r$.
Therefore, all the
upper and lower bound results derived in this paper and in~\cite{eibenesa,eibenjcss} pertaining to
these clustering problems hold true for their corresponding graph
problems on~$\QQQ_d^r$.

\begin{corollary}
\label{cor:graphs}
Given $r,d,k \in \Nat$ and a graph $G\in \QQQ^r_d$, determining whether $G$ has a:
\begin{itemize}
\item dominating set of size $k$ is \FPT\ parameterized by $k+r$;
\item partition into $k$ cliques is \FPT\ parameterized by $k+r$;
\item independent set of size $k$ is \FPT\ parameterized by $k+r$;
\item clique of size $k$ is \FPT\ parameterized by $r$.
\end{itemize}
\end{corollary}

 We note that all the tractability results outlined in
Corollary~\ref{cor:graphs} are tight, which follows from the
lower-bound results obtained in Section~\ref{sec:lb} and in~\cite{eibenesa,eibenjcss}, in the sense that dropping any parameter from our parameterizations leads to an intractable problem.
  
Observing that three of the graph properties in the problems discussed above are expressible in First Order Logic (FO) and result in FO formulas whose length is a function of the parameter $k$, an interesting question that ensues from the above discussion is whether these positive results can be extended to the generic problem of First-Order Model Checking~\cite{Libkin04,GroheKS17}, formalised below. We will show next that the answer to this question is negative---and, in fact, remains negative even when we restrict ourselves to induced subgraphs of hypercubes (\ie, for $r=1$).

\pbDefP{$\QQQ$-\textsc{FO-Model-Checking}}{A first-order (FO) formula
  $\phi$, integers $d,r$, and a graph $G \in \QQQ_d^r$.}{$|\Phi|$}{Does $G \models \Phi$?}
We denote by \textsc{FO-Model-Checking} the general FO Model Checking
problem on graphs, \ie, $\CCC$-\textsc{FO-Model-Checking} with $\CCC$
being the class of all graphs.

\newcommand{\HHH}{\mathcal{H}}
 
\begin{lemma}\label{lem:subdivision}
  Let $H$ be an arbitrary graph. There is a graph $G \in \QQQ^1_{|V(H)|+|E(H)|}$ such that $G$ is isomorphic to the graph $H'$ obtained from $H$ after subdividing every edge of $H$ exactly once and attaching a leaf to every vertex resulting from a subdivision. Moreover, $G$ can be computed from $H$ in polynomial time.
   \end{lemma}
\begin{proof}
 Let $n=|V(H)|$ and $m=|E(H)|$.
To prove the lemma, we construct a matrix representation $M \in \{0,1\}^{n+m}$ of $H'$ which has one row (vector) for every vertex in $H$ and where two vertices in $H'$ are
  adjacent if and only if their corresponding rows in $M$ have Hamming
  distance at most $1$. Let $v_1,\ldots, v_n$ be an arbitrary ordering of the vertices of
  $H$, and $e_1,\ldots,e_m$ be an arbitrary ordering of its edges. Then, $M$ contains one row $r_i$ for every $i \in [n]$ that is $1$
  at its $i$-th entry and $0$ at all other entries. Moreover, for
  every edge $e_\ell=\{v_i,v_j\} \in E(H)$, $M$ contains the following
  two rows:
  \begin{itemize}
  \item the row $r_e$ (corresponding to the degree-3 vertex in $H'$ obtained from $e$) that
    is $1$ at the $i$-th and $j$-th entries, and $0$ at all other entries;
    and
  \item the row $r_e'$ (corresponding to the leaf in $H'$ obtained from $e$) that
    is $1$ at the $i$-th, $j$-th, and ($n+\ell$)-th entries, and $0$ at all other entries.
  \end{itemize}
  This completes the construction of $M$. Clearly, two rows in $M$
  have Hamming distance at most one if and only if their corresponding
  vertices in $H'$ are adjacent, as required. 
\end{proof}
 
\begin{theorem}

  $\QQQ$-\textsc{FO-model-checking} is \Weft\emph{[t]}-hard for
  every $t \in \Nat^*$.
\end{theorem}
\begin{proof}
  We give a parameterized reduction from \textsc{FO Model
    Checking}, which is \Weft[t]-hard for
  every $t \in \Nat^*$. Let $\III:=(\Phi,H)$ be an instance of \textsc{FO Model
    Checking}. We will show the theorem by constructing the equivalent
  instance $\III':=(\Phi',G)$ such that $G \in \QQQ_d^1$ and
  $|\Phi|\leq f(|\Phi'|)$ for some computable function
  $f$ and value $d$ that is polynomially bounded in the input size. $G$ is obtained from $H$ in the same manner as in Lemma~\ref{lem:subdivision}.
     Moreover, $\Phi'$ is obtained
  from $\Phi$ as follows:
  \begin{itemize}
  \item Let $\phi_V(x)$ be the formula that holds for a variable $x$
    if and only if $x$ corresponds to one of the original vertices in
    $G$, \ie, $\phi_V(x):=\forall y E(x,y) \exists z
    \neq x \land E(y,z)$;
  \item replace every subformula of the form $\exists x \phi$ (for some
    variable $x$ and some subformula $\phi$ of $\Phi$) with the
    formula $\exists x \phi_V(x) \land \phi$;
  \item replace every subformula of the form $\forall x \phi$ (for some
    variable $x$ and some subformula $\phi$ of $\Phi$) with the
    formula $\forall x \phi_V(x) \rightarrow \phi$; and
  \item replace every atom $E(x,y)$, where $E$ is the adjacency
    predicate and $x$ and $y$ are variables, with the formula $\exists
    s E(x,s) \land E(s,y)\land x\neq y$.
  \end{itemize}
  It is straightforward now to show that $H \models \Phi$ if and only if
  $G \models \Phi'$, and that $|\Phi'|\leq 20|\Phi|$. Moreover, because of Lemma~\ref{lem:subdivision},
  $G' \in \QQQ_d^1$, as required. 
\end{proof}

\section{Conclusion}
In this paper, we studied the parameterized complexity of the classical {\sc Independent Set} problem on induced subgraphs of powers of hypercubes, but with the additional complication that the ``positions'' of the vertices in the hypercube representation may be partially unknown. We considered the two most natural parameters for the problem: the size $k$ of the independent set and the power $r$ of the hypercube, and provided a complete characterisation of the problem's complexity w.r.t.\ $k$ and $r$.
  We also performed a meta-investigation of the parameterized complexity of graph problems on this graph class that are expressible in FO logic and showed the existence of such problems that are parameterized intractable.

A natural future direction of our work is to study the parameterized complexity of other graph problems on this class, in particular those that have applications in clustering. One famous open problem that comes to mind is the $p$-center problem~\cite{frieze,lingashammingcenter}. The problem can be formulated similarly to the above setting, with the exception of allowing the selection of vertices to be from the whole hypercube, as opposed to restricting them to the input subgraph. In particular, the well-known $p$-centers problem reduces to the $k$-dominating set problem 
in the $r$-th power of the hypercube graph, but where the $k$ vertices in the dominating set are not restricted to the input subgraph, but can be chosen from $\QQQ_d$. This problem was shown to be \FPT{} parameterized by $k+r$~\cite{eibenjcss}. An intriguing \NP-hard restriction of the problem is the problem slice  corresponding to $p=1$, or what is known as the 1-center problem, or equivalently, the {\sc Closest String} problem~\cite{lingashammingcenter,li}. The parameterized complexity of the problem paramertized by each of $k$ and $r$ alone remain important open questions. 

\subparagraph*{Acknowledgments.} The authors wish to thank Subrahmanyam Kalyanasundaram, Suryaansh Jain and Kartheek Sriram Tammana for their help with identifying an error in the previous version of this manuscript.

 \bibliographystyle{plainurl}
\bibliography{literature}

\begin{thebibliography}{10}

\bibitem{clusteringbook1}
Charu~C. Aggarwal and Chandan~K. Reddy.
\newblock {\em Data Clustering: Algorithms and Applications}.
\newblock Chapman \& Hall/CRC, 1st edition, 2013.

\bibitem{fomin1}
Sayan Bandyapadhyay, Fedor~V. Fomin, Petr~A. Golovach, William Lochet, Nidhi
  Purohit, and Kirill Simonov.
\newblock How to find a good explanation for clustering?
\newblock In {\em Thirty-Sixth {AAAI} Conference on Artificial Intelligence,
  {AAAI} 2022, Thirty-Fourth Conference on Innovative Applications of
  Artificial Intelligence, {IAAI} 2022, The Twelveth Symposium on Educational
  Advances in Artificial Intelligence, {EAAI} 2022 Virtual Event, February 22 -
  March 1, 2022}, pages 3904--3912. {AAAI} Press, 2022.

\bibitem{fomin3}
Sayan Bandyapadhyay, Fedor~V. Fomin, Petr~A. Golovach, Nidhi Purohit, and
  Kirill Simonov.
\newblock {FPT} approximation for fair minimum-load clustering.
\newblock In Holger Dell and Jesper Nederlof, editors, {\em 17th International
  Symposium on Parameterized and Exact Computation, {IPEC} 2022, September 7-9,
  2022, Potsdam, Germany}, volume 249 of {\em LIPIcs}, pages 4:1--4:14. Schloss
  Dagstuhl - Leibniz-Zentrum f{\"{u}}r Informatik, 2022.

\bibitem{fomin2}
Sayan Bandyapadhyay, Fedor~V. Fomin, Petr~A. Golovach, Nidhi Purohit, and
  Kirill Simonov.
\newblock Lossy kernelization of same-size clustering.
\newblock In Alexander~S. Kulikov and Sofya Raskhodnikova, editors, {\em
  Computer Science - Theory and Applications - 17th International Computer
  Science Symposium in Russia, {CSR} 2022, Virtual Event, June 29 - July 1,
  2022, Proceedings}, volume 13296 of {\em Lecture Notes in Computer Science},
  pages 96--114. Springer, 2022.

\bibitem{fomin7}
Sayan Bandyapadhyay, Fedor~V. Fomin, Petr~A. Golovach, and Kirill Simonov.
\newblock Parameterized complexity of feature selection for categorical data
  clustering.
\newblock In Filippo Bonchi and Simon~J. Puglisi, editors, {\em 46th
  International Symposium on Mathematical Foundations of Computer Science,
  {MFCS} 2021, August 23-27, 2021, Tallinn, Estonia}, volume 202 of {\em
  LIPIcs}, pages 14:1--14:14. Schloss Dagstuhl - Leibniz-Zentrum f{\"{u}}r
  Informatik, 2021.

\bibitem{fomin6}
Sayan Bandyapadhyay, Fedor~V. Fomin, and Kirill Simonov.
\newblock On coresets for fair clustering in metric and {Euclidean} spaces and
  their applications.
\newblock In Nikhil Bansal, Emanuela Merelli, and James Worrell, editors, {\em
  48th International Colloquium on Automata, Languages, and Programming,
  {ICALP} 2021, July 12-16, 2021, Glasgow, Scotland (Virtual Conference)},
  volume 198 of {\em LIPIcs}, pages 23:1--23:15. Schloss Dagstuhl -
  Leibniz-Zentrum f{\"{u}}r Informatik, 2021.

\bibitem{BhattacharyaH99}
Binay~K. Bhattacharya and Michael~E. Houle.
\newblock Generalized maximum independent sets for trees in subquadratic time.
\newblock In Alok Aggarwal and C.~Pandu Rangan, editors, {\em Algorithms and
  Computation, 10th International Symposium, {ISAAC} '99, Chennai, India,
  December 16-18, 1999, Proceedings}, volume 1741 of {\em Lecture Notes in
  Computer Science}, pages 435--445. Springer, 1999.

\bibitem{cp10}
Emmanuel~J. Cand{\`{e}}s and Yaniv Plan.
\newblock Matrix completion with noise.
\newblock {\em Proceedings of the {IEEE}}, 98(6):925--936, 2010.

\bibitem{cr09}
Emmanuel~J. Cand{\`{e}}s and Benjamin Recht.
\newblock Exact matrix completion via convex optimization.
\newblock {\em Foundations of Computational Mathematics}, 9(6):717--772, 2009.

\bibitem{ct10}
Emmanuel~J. Cand{\`{e}}s and Terence Tao.
\newblock The power of convex relaxation: near-optimal matrix completion.
\newblock {\em {IEEE} Trans. Information Theory}, 56(5):2053--2080, 2010.

\bibitem{CeccarelloPPU17}
Matteo Ceccarello, Andrea Pietracaprina, Geppino Pucci, and Eli Upfal.
\newblock {MapReduce} and streaming algorithms for diversity maximization in
  metric spaces of bounded doubling dimension.
\newblock {\em {PVLDB}}, 10(5):469--480, 2017.

\bibitem{charikar}
Moses Charikar and Rina Panigrahy.
\newblock Clustering to minimize the sum of cluster diameters.
\newblock {\em Journal of Computer and System Sciences}, 68(2):417 -- 441,
  2004.

\bibitem{CyganFKLMPPS15}
Marek Cygan, Fedor~V. Fomin, Lukasz Kowalik, Daniel Lokshtanov, D{\'{a}}niel
  Marx, Marcin Pilipczuk, Michal Pilipczuk, and Saket Saurabh.
\newblock {\em Parameterized Algorithms}.
\newblock Springer, 2015.

\bibitem{DowneyFellows13}
Rodney~G. Downey and Michael~R. Fellows.
\newblock {\em Fundamentals of Parameterized Complexity}.
\newblock Texts in Computer Science. Springer, 2013.

\bibitem{Dvorak2007}
Tom\'{a}\v{s} Dvo\v{r}\'{a}k and Petr Gregor.
\newblock Hamiltonian paths with prescribed edges in hypercubes.
\newblock {\em Discrete Mathematics}, 307(16):1982--1998, 2007.

\bibitem{frieze}
M.E Dyer and A.M Frieze.
\newblock A simple heuristic for the $p$-centre problem.
\newblock {\em Oper. Res. Lett.}, 3(6):285--288, 1985.

\bibitem{fomin8}
Eduard Eiben, Fedor~V. Fomin, Petr~A. Golovach, William Lochet, Fahad Panolan,
  and Kirill Simonov.
\newblock {EPTAS} for \emph{k}-means clustering of affine subspaces.
\newblock In D{\'{a}}niel Marx, editor, {\em Proceedings of the 2021 {ACM-SIAM}
  Symposium on Discrete Algorithms, {SODA} 2021, Virtual Conference, January 10
  - 13, 2021}, pages 2649--2659. {SIAM}, 2021.

\bibitem{EibenGKOS21}
Eduard Eiben, Robert Ganian, Iyad Kanj, Sebastian Ordyniak, and Stefan Szeider.
\newblock The parameterized complexity of clustering incomplete data.
\newblock In {\em Thirty-Fifth {AAAI} Conference on Artificial Intelligence,
  {AAAI} 2021}, pages 7296--7304. {AAAI} Press, 2021.
\newblock URL: \url{https://ojs.aaai.org/index.php/AAAI/article/view/16896}.

\bibitem{eibenesa}
Eduard Eiben, Robert Ganian, Iyad Kanj, Sebastian Ordyniak, and Stefan Szeider.
\newblock Finding a cluster in incomplete data.
\newblock In Shiri Chechik, Gonzalo Navarro, Eva Rotenberg, and Grzegorz
  Herman, editors, {\em 30th Annual European Symposium on Algorithms, {ESA}
  2022, September 5-9, 2022, Berlin/Potsdam, Germany}, volume 244 of {\em
  LIPIcs}, pages 47:1--47:14. Schloss Dagstuhl - Leibniz-Zentrum f{\"{u}}r
  Informatik, 2022.

\bibitem{eibenjcss}
Eduard Eiben, Robert Ganian, Iyad Kanj, Sebastian Ordyniak, and Stefan Szeider.
\newblock On the parameterized complexity of clustering problems for incomplete
  data.
\newblock {\em Journal of Computer and System Sciences}, 134:1--19, 2023.

\bibitem{ev13}
Ehsan Elhamifar and Ren{\'{e}} Vidal.
\newblock Sparse subspace clustering: Algorithm, theory, and applications.
\newblock {\em {IEEE} Trans. Pattern Anal. Mach. Intell.}, 35(11):2765--2781,
  2013.

\bibitem{Erdos60}
Paul Erd{\"o}s and Richard Rado.
\newblock Intersection theorems for systems of sets.
\newblock {\em Journal of the London Mathematical Society}, 1(1):85--90, 1960.

\bibitem{normclustering}
Tom\'{a}s Feder and Daniel Greene.
\newblock Optimal algorithms for approximate clustering.
\newblock In {\em Proceedings of the Twentieth Annual ACM Symposium on Theory
  of Computing}, STOC '88, pages 434--444. ACM, 1988.

\bibitem{FlumGrohe06}
J\"{o}rg Flum and Martin Grohe.
\newblock {\em Parameterized Complexity Theory}, volume XIV of {\em Texts in
  Theoretical Computer Science. An EATCS Series}.
\newblock Springer, Berlin, 2006.

\bibitem{fomin4}
Fedor~V. Fomin, Petr~A. Golovach, Tanmay Inamdar, Nidhi Purohit, and Saket
  Saurabh.
\newblock Exact exponential algorithms for clustering problems.
\newblock In Holger Dell and Jesper Nederlof, editors, {\em 17th International
  Symposium on Parameterized and Exact Computation, {IPEC} 2022, September 7-9,
  2022, Potsdam, Germany}, volume 249 of {\em LIPIcs}, pages 13:1--13:14.
  Schloss Dagstuhl - Leibniz-Zentrum f{\"{u}}r Informatik, 2022.

\bibitem{fomin5}
Fedor~V. Fomin, Petr~A. Golovach, and Kirill Simonov.
\newblock Parameterized \emph{k}-clustering: Tractability island.
\newblock {\em J. Comput. Syst. Sci.}, 117:50--74, 2021.

\bibitem{HararyHayessurvey}
John P.~Hayes Frank~Harary and Horng-Jyh Wu.
\newblock A survey of the theory of hypercube graphs.
\newblock {\em Comput. Math. Appl.}, 15(4):277--289, 1988.

\bibitem{clusteringbook2}
Guojun Gan, Chaoqun Ma, and Jianhong Wu.
\newblock {\em Data clustering - theory, algorithms, and applications.}
\newblock SIAM, 2007.

\bibitem{ganian}
Robert Ganian, Thekla Hamm, Viktoriia Korchemna, Karolina Okrasa, and Kirill
  Simonov.
\newblock The complexity of k-means clustering when little is known.
\newblock In Kamalika Chaudhuri, Stefanie Jegelka, Le~Song, Csaba Szepesvari,
  Gang Niu, and Sivan Sabato, editors, {\em Proceedings of the 39th
  International Conference on Machine Learning}, volume 162 of {\em Proceedings
  of Machine Learning Research}, pages 6960--6987, 2022.

\bibitem{icml}
Robert Ganian, Iyad Kanj, Sebastian Ordyniak, and Stefan Szeider.
\newblock Parameterized algorithms for the matrix completion problem.
\newblock In {\em {ICML}}, volume~80 of {\em {JMLR} Workshop and Conference
  Proceedings}, pages 1642--1651, 2018.

\bibitem{gawrychowski_et_al}
Pawel Gawrychowski, Nadav Krasnopolsky, Shay Mozes, and Oren Weimann.
\newblock {Dispersion on Trees}.
\newblock In Kirk Pruhs and Christian Sohler, editors, {\em 25th Annual
  European Symposium on Algorithms (ESA 2017)}, volume~87 of {\em Leibniz
  International Proceedings in Informatics (LIPIcs)}, pages 40:1--40:13.
  Schloss Dagstuhl--Leibniz-Zentrum fuer Informatik, 2017.

\bibitem{lingashammingcenter}
Leszek G\c{a}sieniec, Jesper Jansson, and Andrzej Lingas.
\newblock Efficient approximation algorithms for the {H}amming center problem.
\newblock In {\em Proceedings of the Tenth Annual ACM-SIAM Symposium on
  Discrete Algorithms}, pages 905--906, 1999.

\bibitem{lingasapxclustering}
Leszek G\c{a}sieniec, Jesper Jansson, and Andrzej Lingas.
\newblock Approximation algorithms for {Hamming} clustering problems.
\newblock {\em Journal of Discrete Algorithms}, 2(2):289 -- 301, 2004.

\bibitem{gonzalez}
Teofilo~F. Gonzalez.
\newblock Clustering to minimize the maximum intercluster distance.
\newblock {\em Theoretical Computer Science}, 38:293 -- 306, 1985.

\bibitem{GrammNiedermeierRossmanith03}
Jens Gramm, Rolf Niedermeier, and Peter Rossmanith.
\newblock Fixed-parameter algorithms for {CLOSEST} {STRING} and related
  problems.
\newblock {\em Algorithmica}, 37(1):25--42, 2003.

\bibitem{GroheKS17}
Martin Grohe, Stephan Kreutzer, and Sebastian Siebertz.
\newblock Deciding first-order properties of nowhere dense graphs.
\newblock {\em J. {ACM}}, 64(3):17:1--17:32, 2017.

\bibitem{hmrw14}
Moritz Hardt, Raghu Meka, Prasad Raghavendra, and Benjamin Weitz.
\newblock Computational limits for matrix completion.
\newblock In {\em Proceedings of The 27th Conference on Learning Theory},
  volume~35 of {\em {JMLR} Workshop and Conference Proceedings}, pages
  703--725. JMLR.org, 2014.

\bibitem{hermelin}
Danny Hermelin and Liat Rozenberg.
\newblock Parameterized complexity analysis for the closest string with
  wildcards problem.
\newblock {\em Theoretical Computer Science}, 600:11--18, 2015.

\bibitem{niedermeierradius}
Tomohiro Koana, Vincent Froese, and Rolf Niedermeier.
\newblock Parameterized algorithms for matrix completion with radius
  constraints.
\newblock In Inge~Li G{\o}rtz and Oren Weimann, editors, {\em 31st Annual
  Symposium on Combinatorial Pattern Matching, {CPM} 2020, June 17-19, 2020,
  Copenhagen, Denmark}, volume 161 of {\em LIPIcs}, pages 20:1--20:14. Schloss
  Dagstuhl - Leibniz-Zentrum f{\"{u}}r Informatik, 2020.

\bibitem{niedermeierdiameter}
Tomohiro Koana, Vincent Froese, and Rolf Niedermeier.
\newblock The complexity of binary matrix completion under diameter
  constraints.
\newblock {\em J. Comput. Syst. Sci.}, 132:45--67, 2023.

\bibitem{clusteringbook4}
Jure Leskovec, Anand Rajaraman, and Jeffrey~David Ullman.
\newblock {\em Mining of Massive Datasets}.
\newblock Cambridge University Press, New York, NY, USA, 2nd edition, 2014.

\bibitem{li}
Ming Li, Bin Ma, and Lusheng Wang.
\newblock On the closest string and substring problems.
\newblock {\em J. ACM}, 49(2):157--171, 2002.

\bibitem{Libkin04}
Leonid Libkin.
\newblock {\em Elements of Finite Model Theory}.
\newblock Texts in Theoretical Computer Science. An {EATCS} Series. Springer,
  2004.

\bibitem{clusteringbook3}
Boris Mirkin.
\newblock {\em Clustering For Data Mining: A Data Recovery Approach}.
\newblock Chapman \& Hall/CRC, 2005.

\bibitem{SacharidisMSPV18}
Dimitris Sacharidis, Paras Mehta, Dimitrios Skoutas, Kostas Patroumpas, and
  Agn{\`{e}}s Voisard.
\newblock Selecting representative and diverse spatio-textual posts over
  sliding windows.
\newblock In {\em Proceedings of the 30th International Conference on
  Scientific and Statistical Database Management, {SSDBM} 2018, Bozen-Bolzano,
  Italy, July 09-11, 2018}, pages 17:1--17:12, 2018.

\end{thebibliography}
  
\end{document}